\documentclass[final,3p,10pt]{elsarticle}
\usepackage{hyperref}

\journal{Neural Networks}

\usepackage{float, amsmath, mathtools, amssymb}
\input{Qcircuit}
\usepackage{algorithm,algorithmic}
\usepackage{graphicx,color,tikz}

\newcommand{\field}[1]{\mathbb{#1}}
\newcommand{\vs}[1]{\mathbb{#1}}
\newcommand{\qop}[1]{\textsf{#1}}

\newtheorem{theorem}{Theorem}[section]
\newtheorem{proposition}[theorem]{Proposition}

\newenvironment{proof}[1][Proof]{\begin{trivlist}
\item[\hskip \labelsep {\bfseries #1}]}{\end{trivlist}}
\newenvironment{definition}[1][Definition]{\begin{trivlist}
\item[\hskip \labelsep {\bfseries #1}]}{\end{trivlist}}

\begin{document}

\begin{frontmatter}

\title{Quantum perceptron over a field and neural network architecture selection in a quantum computer}

\author[mymainaddress,mysecondaryaddress]{Adenilton Jos\'e da Silva}
\cortext[mycorrespondingauthor]{A. J. da Silva}
\ead{adenilton.silva@ufrpe.br}

\author[mysecondaryaddress]{Teresa Bernarda Ludermir}

\author[mymainaddress]{Wilson Rosa de Oliveira}


\address[mymainaddress]{Departamento de Estat\'{i}stica e Inform\'{a}tica\\ Universidade 
Federal Rural de Pernambuco}
\address[mysecondaryaddress]{Centro de Inform\'{a}tica, Universidade Federal de Pernambuco\\
Recife, Pernambuco, Brasil}

\begin{abstract}
In this work, we propose a quantum neural network named quantum perceptron over a field (QPF).  Quantum computers are not yet a reality and the models and algorithms proposed in this work cannot be simulated in actual (or classical) computers. QPF is a direct generalization of a classical perceptron and solves some drawbacks found in previous models of quantum perceptrons. We also present a learning algorithm named Superposition based Architecture Learning algorithm (SAL) that optimizes the neural network weights and architectures. SAL searches for the best architecture in a finite set of neural network architectures with linear time over the number of patterns in the training set. SAL is the first learning algorithm to determine neural network architectures in polynomial time. This speedup is obtained by the use of quantum parallelism and a non-linear quantum operator.
\end{abstract}


\end{frontmatter}


\section{Introduction}
\label{sec:introduction}

The size of computer components reduces each year and quantum effects have to be eventually considered in computation with future hardware. The theoretical possibility of quantum computing initiated with Benioff~\cite{benioff} and Feynman~\cite{Feynman} and the formalization of the first quantum computing model was proposed by Deutsch in 1985~\cite{deutsch1985quantum}. The main advantage of quantum computing over classical computing is the use of a principle called superposition which allied with the linearity of the operators allows for a powerful form of parallelism  to develop algorithms more efficient than the known classical ones. For instance, the Grover's search algorithm~\cite{grover:96} and Shor's factoring algorithm~\cite{shor:97} overcome any known classical algorithm.

Quantum computing has recently been used in the development of new machine learning techniques as quantum decision trees~\cite{Farhi1998}, artificial neural networks~\cite{Narayanan,panella:09,daSilva:12}, associative 
memory~\cite{ventura:98,PhysRevLett.87.067901}, and inspired the development of novel evolutionary algorithms for continuous optimization problems~\cite{Hsu_2013,Duan2010}. There is an increasing interest in quantum machine learning and in the quantum neural network area~\cite{Schuld2014}. This paper proposes a quantum neural network named Quantum Perceptron over a Field (QPF) and investigates the use of quantum computing techniques to design a learning algorithm for neural networks. Empirical evaluation of QPF and its learning algorithm needs of a quantum computer with thousands of qubits. Such quantum computer is not available nowadays and an empirical analysis of the QPF and its learning algorithm is not possible with current technology.

Artificial neural networks are a universal model of computation~\cite{doi:10.1142/S0129065714500294} and have several applications in real life problems. For instance, in the solution of combinatorial optimization problems~\cite{doi:10.1142/S0129065714400061},  pattern recognition~\cite{doi:10.1142/S0129065714500257}, 
but have some problems as the lack of an algorithm to determine optimal architectures~\cite{doi:10.1142/S0129065703001467}, memory capacity and high cost learning algorithms~\cite{doi:10.1142/S0129065701000655}.

Notions of Quantum Neural Networks have been put forward since the nineties~\cite{kak:95}, but a precise definition of what is a quantum neural network that integrates neural computation and quantum computation is a non-trivial open problem~\cite{Schuld2014}.
To date, the proposed models in the literature are really just quantum inspired in the sense that despite using a quantum representation of data and quantum operators, in a way or another some quantum principles are violated usually during training. Weights adjustments need measurements (observation) and updates.

Research in quantum neural computing is unrelated, as stated in~\cite{Schuld2014}: 
\begin{quote}
``QNN research remains an exotic conglomeration of different ideas under the umbrella of quantum information''.
\end{quote}
and there is no consensus of what are the components of a quantum neural network. Several models of quantum neural networks have been proposed and they present different conceptual models. In some models a quantum  neural network is described as a physical  device~\cite{Narayanan}; as a model only inspired in quantum computing~\cite{kouda:05}; or as a mathematical model that explores quantum computing principles~\cite{zhou:07,panella:09,daSilva:12,Schuld2014walks}. We follow the last approach and assume that our quantum neural network model would be implemented in a quantum computer that follows the quantum principles as \emph{e.g.} described in~\cite{nielsen:00}. We assume that our models are implemented in the quantum circuit  model of Quantum Computing~\cite{nielsen:00}.

Some advantages of quantum neural models over the classical models are the exponential gain in memory capacity~\cite{trugenberger:02}, quantum neurons can solve nonlinearly separable problems~\cite{zhou:07}, and a nonlinear quantum learning algorithm with polynomial time over the number of patterns in the data set is presented in~\cite{panella:09}. However, these quantum neural models cannot be viewed as a direct generalization of a classical neural network and have some limitations presented in Section~\ref{sec:related_works}. Quantum computing simulation has exponential cost in relation to the number of qubits. Experiments with benchmarks and real problems are not possible because of the number of qubits necessary to simulate a quantum neural network.

The use of artificial neural networks to solve a problem requires considerable time for choosing parameters and neural network architecture~\cite{Almeida2010}. The architecture design is extremely important in neural network applications because a neural network with a simple architecture may not be capable of performing the task. On the other hand, a complex architecture can overfit the training  data~\cite{doi:10.1142/S0129065703001467}. The definition of an algorithm to determine (in a finite set of architectures) the best neural network architecture (minimal architecture  for a given learning task that can learn the training dataset) efficiently is an open problem. The objective of this paper is to show that with the supposition of non-linear quantum computing~\cite{panella:09,panella2009neurofuzzy,PhysRevLett.81.3992} we can determine an architecture that can learn the training data in linear time with relation to the number of patterns in the training set. To achieve this objective, we propose a quantum neural network that respect the principles of quantum computation, neural computing and generalizes the classical perceptron. The proposed neuron works as a classical perceptron when the weights are in the computational basis, but as quantum perceptron when the weights are in superposition. We propose a neural network learning algorithm which uses a non-linear quantum operator~\cite{panella:09,PhysRevLett.81.3992} to perform a global search in the space of weights and architecture of a neural network. The proposed learning algorithm is the first quantum algorithm performing this kind of optimization in polynomial time and presents a framework to develop linear quantum learning algorithms to find near optimal neural network architectures.

 The remainder of this paper is divided into 6 Sections. In Section~\ref{sec:scope} we describe models that are out of the scope of this work. In Section ~\ref{sec:qc} 
we present preliminary concepts of quantum computing necessary to understand this work.  In Section~\ref{sec:related_works} we present related works. Section~\ref{sec:quantum_neuron} presents main results of this paper. We define the new model of a quantum neuron named quantum perceptron over a field that respect principles of quantum and neural computing. Also in Section~\ref{sec:quantum_neuron} we propose a quantum learning algorithm for neural networks that determines a neural network architecture that can learn the train set with some desired accuracy. Section~\ref{sec:discussion} presents a discussion. Finally, Section~\ref{sec:conclusion} is the conclusion.

\section{Out of scope}
\label{sec:scope}

Quantum computing and neural networks are multidisciplinary research fields. In 
this way, the quantum neural computing research is also multidisciplinary and 
concepts from physics, mathematics and computer science are used. Probably 
because of this multidisciplinary characteristic there are completely different 
concepts named quantum neural networks. In this Section, we point some models 
that are out of the scope of this work.
 
\subsection{Quantum inspired neural networks} 
Neural networks whose definition is based on quantum computation, but that 
works in a classical computer as in \cite{kouda:05,Zhou2006,Li201381}  are named 
in this work as Quantum Inspired Neural Networks. Quantum inspired neural 
networks are not real quantum models. Quantum inspired models are classical 
neural networks that are inspired in quantum computing exactly as there are 
combinatorial optimization methods inspired in ant colony or bird swarm. 

In~\cite{kouda:05} a complex neural network named qubit neural network whose 
neurons acts in the phase of the input values is proposed. The qubit neural 
network has its functionality based in quantum operation, but it is a classical 
model and can be efficiently simulated in a classical computer.

Another quantum inspired models is defined in~\cite{zhou99} where the activation 
function is a linear combination of sigmoid functions. This linear combination 
of activation functions is inspired in the concept of quantum superposition, but 
these models can be efficiently simulated by a classical computer.

Quantum inspiration can bring useful new ideas and techniques for neural network 
models and learning algorithms design. However, quantum inspired neural networks 
are  out of the scope of this work.

\subsection{Physical device quantum neural networks} 

Devices that implement a quantum neural network are proposed in 
\cite{Narayanan,Behrman}. In this work, these models are named physical device 
quantum neural network. The main problem of this kind of proposal is the 
hardware dependence. Scalable quantum computers are not yet a reality and when 
someone build a quantum computer we do not know which hardware will be used.

 In~\cite{Narayanan} a quantum neural network is represented by the architecture 
of a double slit experiment where input patterns are represented by photons, 
neurons are represented by slits, weights are represented by waves and screen 
represents output neurons. In~\cite{Behrman} a quantum neural network is 
represented by a quantum dot molecule evolving in real time. Neurons are 
represented by states of molecules, weights are the number of excitations that 
are optically controlled, inputs are the initial state of the quantum dot 
molecules and outputs are the final state of the dot molecules. 

Physical device quantum neural networks are real quantum models. This kind of 
quantum neural networks needs of specific hardware and is out of the scope of 
this work.

\subsection{Quantum inspired algorithms} 
In this work algorithms whose development uses ideas from quantum computing, but 
run in a classical computer are named quantum inspired algorithms. For instance, 
there are several quantum inspired evolutionary algorithm proposed in 
literature~\cite{kasabov2011quantum,platel2009quantum,han2002quantum,tirumala2014quantum,vellasco,vellasco2}.
This kind of algorithm uses quantum inspiration to define better classical 
algorithms, but intrinsic quantum properties are not used. Quantum inspired 
algorithms are out of the scope of this work.

\section{Quantum computing}
\label{sec:qc}
In this Section, we perform a simple presentation of quantum computing with the necessary concepts to understand the following sections. As in theoretical classical computing we are not interested in how to store or physically represent a quantum bit. Our approach is a mathematical. We deal with  how we can abstractly compute with quantum bits or design abstract quantum circuits and models. We take as a guiding principle that when a universal quantum computer will be at our disposal, we could implement the proposed quantum neural models.

We define a quantum bit as unit complex bi-dimensional vector in the vector space $\field{C}^2$. In the quantum computing literature a vector is represented by the Dirac’s notation
 $\ket{\cdot}$. The computational basis is the set $\left\{\ket{0},\ket{1}\right\}$, 
 where the vectors $\ket{0}$ and $\ket{1}$ can be represented as in Equation~\eqref{eq:compbasis}. 
 
 \begin{equation}
 \ket{0} = \begin{bmatrix}1 \\ 0\end{bmatrix} \mbox{ and } \ket{1} = \begin{bmatrix}0 \\ 1\end{bmatrix}
 \label{eq:compbasis}
 \end{equation}
 A quantum bit $\ket{\psi}$, \emph{qubit}, 
 is a vector in $\field{C}^2$ that has a unit length as described in Equation~\eqref{eq:qbit}, 
 where $|\alpha|^2+|\beta|^2=1$, $\alpha,\beta\in\field{C}$.
 \begin{equation}
 \ket{\psi} = \alpha\ket{0} + \beta \ket{1}
 \label{eq:qbit}
 \end{equation}
 
 While in classical computing there are only two possible bits, 0 or 1,  in quantum computing there are an infinitely many  quantum bits,  a quantum bit can be a linear combination (or superposition) of $\ket{0}$ and $\ket{1}$. 
The inner product of two qubits $\ket{\psi}$ and $\ket{\theta}$ is denoted as $\langle\psi|\theta\rangle$ and the symbol $\bra{\psi}$ is the transpose conjugate of the vector $\ket{\psi}$.
 
 Tensor product $\otimes$ is used to define multi-qubit systems. On the vectors, if
$\ket{\psi}=\alpha_0 \ket{0}+\alpha_1 \ket{1}$ and $\ket{\phi}=\beta_0 
\ket{0}+\beta_1 \ket{1}$, then the tensor product $\ket{\psi}\otimes \ket{\phi}$ is equal to the vector  $\ket{\psi \phi}=\ket{\psi} \otimes \ket{\phi} =$ 
$ (\alpha_0 \ket{0}+\alpha_1 \ket{1})\otimes (\beta_0 \ket{0}+\beta_1 
\ket{1})=\alpha_0 \beta_0 \ket{00}+ \alpha_0 \beta_1 \ket{01}+\alpha_1 \beta_0 
\ket{10}+\alpha_1 \beta_1 \ket{11}$. Quantum states representing $n$ qubits are in a
$2^n$ dimensional complex vector space. On the spaces, let $\field{X}\subset\field{V}$ and $\field{X}'\subset\field{V}'$ be basis of respectively vector spaces $\field{V}$ and $\field{V}'$. The tensor product $\field{V}\otimes \field{V}'$  is the vector space obtained from the basis
$\left\{ \ket{b} \otimes \ket{b'}; \ket{b} \in \field{X} \mbox{ and } \ket{b'}\in \field{X}' \right\}$. 
The symbol $\field{V}^{\otimes n}$ represents a tensor product   
$\field{V}\otimes \cdots \otimes \field{V}$ with $n$ factors.

 Quantum operators over $n$ qubits are represented by $2^n \times 2^n$ unitary matrices. An $n\times n$ matrix $\qop{U}$ is unitary if $\qop{U}\qop{U}^\dag=\qop{U}^\dag\qop{U}=\qop{I}$, where $\qop{U}^\dag$ is the transpose conjugate of $\qop{U}$. For instance, identity $\qop{I}$, the flip $\qop{X}$, and the Hadamard $\qop{H}$ operators are important quantum operators over one qubit and they are described in Equation~\eqref{eq:qo}. 
\begin{equation}
\begin{array}{lll}
\qop{I} = \begin{bmatrix}
1 & 0\\
0 & 1
\end{bmatrix}
&
\qop{X} = \begin{bmatrix}
0 & 1\\
1 & 0
\end{bmatrix}
&
\qop{H} = \frac{1}{\sqrt{2}}\begin{bmatrix}
1 & 1\\
1 & -1
\end{bmatrix}
\end{array}
\label{eq:qo}
\end{equation}
The operator described in Equation~\eqref{eq:cnot} is the \emph{controlled not} operator $\qop{CNot}$,  that flips the second qubit if the first (the controlled qubit) is the state $\ket{1}$. 

\begin{equation}
\qop{CNot} = \begin{bmatrix}
1 & 0 & 0 & 0 \\
0 & 1 & 0 & 0 \\
0 & 0 & 0 & 1 \\
0 & 0 & 1 & 0 \\
\end{bmatrix}
\label{eq:cnot}
\end{equation}

 Parallelism is one of the most important properties of quantum computation used in this paper. If $f:\mathbb{B}^m \rightarrow \mathbb{B}^n$ is a Boolean function, $\mathbb{B}=\{0,1\}$, one can define a quantum operator 
\begin{equation}
\qop{U}_f:\left(\mathbb{C}^2\right)^{\otimes n+m} \rightarrow \left(\mathbb{C}^2\right)^{\otimes n+m},
\label{qof}
\end{equation}
 as
  \[\qop{U}_f \ket{x,y}=\ket{ x,y\oplus f(x)},\]
where $\oplus$ is the bitwise $\qop{xor}$, such that the value of $f(x)$ can be recovered as 
 \[\qop{U}_f \ket{x,0}=\ket{ x,f(x)}.\] 
The operator $\qop{U}_f$ is sometimes called the \emph{quantum oracle} for $f$. Parallelism occurs when one applies the operator $\qop{U}_f$ to a state in superposition as e.g. described in Equation~\eqref{eq:paral}. 
\begin{equation}
\qop{U}_f \left( \sum_{i=0}^n \ket{x_i,0}\right) = \sum_{i=0}^n \ket{x_i,f(x_i)}
\label{eq:paral}
\end{equation}
The meaning of  Equation~\eqref{eq:paral} is that if one applies operator
$\qop{U}_f$ to a state in superposition, by linearity, the value of $f(x_i)$, will be calculated for each $i$ simultaneously with only one single application of the quantum operator $\qop{U}_f$.   

 With quantum parallelism one can imagine that if a problem can be described with a Boolean function $f$ then it can be solved instantaneously. The problem is that one cannot direct access the values in a quantum superposition. In quantum computation a measurement returns a more limited value. The measurement of a quantum state  $\ket{\psi}=\sum_{i=1}^n \alpha_i\ket{x_i}$ in superposition will return $x_i$ with probability $|\alpha_i |^2$. With this result the state $\ket{\psi}$ collapses to state $\ket{x_i}$, i.e., after the measurement 
$\ket{\psi}=\ket{x_i}$.

\section{Related works}
\label{sec:related_works}
 Quantum neural computation research started in the nineties~\cite{kak:95}, and the models (such as \emph{e.g.} in~\cite{Narayanan,panella:09,zhou:07,Altaisky,daSilva:12,ventura:04,Liu2013144,Anguita2003763}) 
are yet unrelated. We identify two types of  quantum neural networks: 1)  models described mathematically to work on a quantum  computer (as in~\cite{Altaisky,ventura:04,zhou:07,panella:09,xuan:11,daSilva:12,Sagheer2013,Siomau2014});  2)  models described by a quantum physical device (as in~\cite{Narayanan,Behrman}). In the following subsections we describe some models of quantum neural networks.

\subsection{qANN}
 Altaisky~\cite{Altaisky} proposed a quantum perceptron (qANN). The 
qANN $N$ is described as in Equation~\eqref{eq:qann}, 
\begin{equation}
\ket{y} = \hat{F}\sum_{j=1}^n \hat{w}_j\ket{x_j}
\label{eq:qann}
\end{equation}
where $\hat{F}$ is a quantum operator over 1 qubit representing the activation function, $\hat{w}_j$ is a quantum operator over a single qubit representing the $j$-th weight of the neuron and $\ket{x_j}$ is one qubit representing the input associated with 
$\hat{w}_j$. 

 The qANN is one of the first models of quantum neural networks. It suggests a way to implement the activation function that is applied (and detailed) for instance in~\cite{panella:09}. 

 It was not described in~\cite{Altaisky} how one can implement 
Equation~\eqref{eq:qann}. An algorithm to put patterns in superposition is necessary. For instance, the storage algorithm of a quantum associative memory~\cite{trugenberger:02} can be used to create the output of the qANN. But this kind of algorithm works only with orthonormal states, as shown in Proposition~\ref{proposition:21}.

\begin{proposition}
Let $\ket{\psi}$ and $\ket{\theta}$ be two qubits with probability amplitudes 
in $\mathbb{R}$, if $\frac{1}{\sqrt{2}} \left(\ket{\psi}+ \ket{\theta}\right)$ 
is a unit  vector then $\ket{\psi}$ and $\ket{\theta}$ are 
orthogonal vectors.
\label{proposition:21}
\end{proposition}

\begin{proof}
Let $\ket{\psi}$ and $\ket{\theta}$ be qubits and suppose that $\frac{1}{\sqrt{2}} \left(\ket{\psi}+ \ket{\theta}\right)$ is a unit vector. Under these conditions 
\begin{equation*}
\begin{split}
\frac{1}{2}\left(\ket{\psi}+\ket{\theta},\ket{\psi}+\ket{\theta}\right) = 1 \Rightarrow \\
\frac{1}{2} \left(\left< \psi|\psi \right> +  \left<\psi|\theta\right>+ \left<\theta|\psi\right>+\left<\theta|\theta\right>\right) =1 
\end{split}
\end{equation*}
Qubits are unit vectors, then
\begin{equation}
\frac{1}{2} \left(2 +  2\left<\psi|\theta\right>\right) =1 
\end{equation}
and $\ket{\psi}$ and $\ket{\theta}$ must be orthogonal vectors.
\end{proof}

A learning rule for the qANN has been proposed and it is shown that the learning rule drives the perceptron to the desired state $\ket{d}$~\cite{Altaisky} in the particular case described in Equation~\eqref{eq:learningaltaisky} where is supposed that $F=I$. But this learning rule does not preserve unitary operators~\cite{daSilva:12}.

\begin{equation}
\hat{w}_j(t+1) = \hat{w}_j(t) + \eta \left( \ket{d}-\ket{y(t)}\right)\bra{x_j}
\label{eq:learningaltaisky}
\end{equation}

In~\cite{Sagheer2013} a quantum perceptron named Autonomous Quantum Perceptron Neural Network (AQPNN) is proposed. This model has a learning algorithm that can learn a problem in a small number of iterations when compared with qANN and these weights are represented in a quantum operator as the qANN weights.

\subsection{qMPN}
Proposition~\ref{proposition:21} shows that with the supposition of unitary evolution, qANN with more than one input is not a well defined quantum operator. Zhou and Ding proposed a new kind of quantum perceptron~\cite{zhou:07} which they called  as quantum M-P neural network (qMPN).  The weights of qMPN are stored in a single squared matrix W that represents a quantum operator. We can see the qMPN as a generalized single weight qANN, where the weight is a quantum operator over any number of qubits. The qMPN is described in Equation~\eqref{eq:qmpn}, where $\ket{x}$ is an input with $n$ qubits, $W$ is a quantum operator over $n$ qubits and $\ket{y}$ is the output. 

\begin{equation}
\ket{y} = W \ket{x}
\label{eq:qmpn}
\end{equation}

 They also proposed a quantum learning algorithm for the new model. The learning algorithm for qMPN is described in Algorithm~\ref{alg:learningqmp}. The weights update rule of Algorithm~\ref{alg:learningqmp} is described in Equation~\eqref{eq:learningQMPN},
\begin{equation}
w_{ij}^{t+1} = w_{ij}^t + \eta \left(\ket{d}_i - \ket{y}_i\right)\ket{x}_j
\label{eq:learningQMPN}
\end{equation}
 where $w_{ij}$ are the entries of the matrix 
$W$ 
and $\eta$ is a learning rate.

\begin{algorithm}[H]
    \begin{algorithmic}[1]
        \STATE Let $W(0)$ be a weight matrix\\
        \STATE Given a set of quantum examples in the form 
        $\left(\ket{x},\ket{d}\right)$, 
        where $\ket{x}$ is an input and $\ket{d}$ is the desired output\\
        \STATE
        Calculate $\ket{y}=W(t) \ket{x}$, where $t$ is the iteration number\\
        \STATE Update the weights following the learning rule described in 
        Equation~\eqref{eq:learningQMPN}.\\
        \STATE Repeat steps 3 and 4 until a stop criterion is met.
    \end{algorithmic}
    \caption{Learning algorithm qMPN}
    \label{alg:learningqmp}
\end{algorithm}

The qMPN model has several limitations in respect to the principles of quantum computing. qMPN is equivalent to a single layer neural network and its learning algorithm leads to non unitary neurons as shown in~\cite{daSilva:14}.

\subsection{Neural network with quantum architecture}
 In the last subsections the neural network weights are represented by quantum operators and the inputs are represented by qubits. In the classical case, inputs and free parameters are real numbers. So one can consider to use qubits to represent inputs and weights. This idea was used, for instance, in~\cite{ventura:04,panella:09}. In~\cite{ventura:04} a detailed description of the quantum neural network is not presented.

 In~\cite{panella:09}  a Neural Network with Quantum Architecture (NNQA) based on a complex valued neural network named \emph{qubit neural network}~\cite{kouda:05} is proposed. Qubit neural network is not a quantum neural network being  just inspired in quantum computing. Unlike previous models~\cite{Altaisky,zhou:07}, NNQA uses fixed quantum operators and  the neural network configuration is represented by a string of qubits. This approach is very similar to the classical case, where the neural network configuration for a given architecture is a string of bits. 

 Non-linear activation functions are included in NNQA in the following way. Firstly is performed a discretization of the input and output space, the scalars are represented by Boolean values. In doing so a neuron is just a Boolean function 
$f$ and the quantum oracle  operator for $f$, $\qop{U}_f$, is used to implement the function $f$ acting on the computational basis.

 In the NNQA all the data are quantized with $N_q$ bits. A synapses of the NNQA is a Boolean function $f_0:B^{2N_q}\rightarrow B^{N_q}$. A synapses of the NNQA is a Boolean function described in Equation~\eqref{eq:qaansynapsis}.

\begin{equation}
z = arctan\left( \frac{\sin(y) + \sin(\theta)}{\cos(y) + \cos(\theta)}\right)
\label{eq:qaansynapsis}
\end{equation}

The values $y$ and $\theta$ are angles in the range $\left[-\pi/2,\pi/2\right]$ 
representing the argument of a complex number, which are quantized as described in Equation~\eqref{eq:qanndiscretdata}. The representation of the angle 
$\beta$ is the binary representation of the integer $k$.
\begin{equation}
\beta = \pi \left( -0.5 + \frac{k}{2^{Nq}-1}\right), k = 0, 1, \cdots, 2^{Nq}-1
\label{eq:qanndiscretdata}
\end{equation}

\begin{proposition}
The set $F=\left\{\beta_k |\beta_k=\pi\cdot\left(-0.5+\right.\right.$ 
$\left.\left. k/(2^{Nq}-1)\right),k=0,\cdots,2^{Nq}-1\right\}$  with canonical 
addition and multiplication is not a field.
\label{proposition:qannfield}
\end{proposition}
\begin{proof}
We will only show that the additive neutral element is not in the set. Suppose that $0\in F$. So, 
$$\pi\cdot\left(-0.5+\frac{k}{2^{Nq}-1}\right) = 0\Rightarrow -0.5+\frac{k}{2^{Nq}-1} = 0$$
$$\Rightarrow \frac{k}{2^{Nq}-1}=0.5 \Rightarrow k = (2^{Nq}-1)\cdot \left(-\frac{1}{2}\right)$$
$Nq$ is a positive integer, and $2^{Nq}-1$ is an odd positive integer, then 
$(2^{Nq}-1)\cdot \left(-\frac{1}{2}\right)~\notin~\mathbb{Z}$ which contradicts the assumption 
that $k \in \mathbb{Z}$ and so $F$ is not a field since $0 \notin F$.

\end{proof}

From Proposition~\ref{proposition:qannfield} we conclude that the we  cannot directly lift the operators and algorithms from classical neural networks to NNQA. In a weighted neural network inputs and parameters are rational, real or complex numbers and the set of possible weights of NNQA under operations defined in NNQA neuron is not a field.

\section{Quantum neuron}
\label{sec:quantum_neuron}

\subsection{Towards a quantum perceptron over a field}
 In Definition~\ref{def:neuron} and Definition~\ref{def:nn} we have, respectively, an artificial neuron and a classical neural network as in~\cite{haykin:99}. 
Weights and inputs in classical artificial neural networks normally are in the set of real (or complex) numbers.
\begin{definition}
An \emph{artificial neuron} is described by the Equation~\eqref{eq:an}, where 
$x_1,x_2,\cdots,x_k$ are input signals and $w_1,w_2,\cdots,w_k$ are weights of 
the synaptic links, $f(\cdot)$ is a nonlinear activation function, and $y$ is the output signal of the neuron.
\begin{equation}
y = f\left(\sum_{j=0}^m w_{j}x_j\right)
\label{eq:an}
\end{equation}
\label{def:neuron}
\end{definition}

 In both qANN and qMPN artificial neurons, weights and inputs are in different sets (respectively in quantum operators and qubits) while weights and inputs in a classical perceptron are elements of the same field.  The NNQA model defines an artificial neuron where weights and inputs are strings of qubits. The neuron is based on a complex valued network and does not exactly follow Definition~\ref{def:neuron}. The main problem in NNQA is that the inputs and weights do not form a field with sum and multiplication values as we show in Proposition~\ref{proposition:qannfield}. There is no guarantee that  the set of discretized parameters is closed under the  operations between qubits.  

 Other models of neural networks where inputs and parameters are qubits were presented in~\cite{oliveira:08,silva:10,daSilva:12}. These models are a generalization of weightless neural network models, whose definitions are not similar to Definition~\ref{def:neuron}.

\begin{definition} 
 A \emph{neural network} is a directed graph consisting of nodes with 
interconnecting synaptic and activation links, and is characterized by four 
properties~\cite{haykin:99}: 
\begin{enumerate}
 \item Each neuron is represented by a set of linear synaptic links, an 
externally applied bias, and a possibility non-linear activation link. The bias 
is represented by a synaptic link connected to an input fixed at +1.
 \item The synaptic links of a neuron weight their respective input signals.
 \item The weighted sum of the input signals defines the induced local field of 
the neuron in question.
 \item The activation link squashes the induced local field of the neuron to 
produce an output. 
\end{enumerate}
\label{def:nn}
\end{definition}

  The architecture of NNQA can be viewed as a directed graph consisting of nodes with interconnecting synaptic and activation links as stated in Definition~\ref{def:nn}. The NNQA does not follow all properties of a neural network (mainly because it is based in the qubit NN), but it is one of the first quantum neural networks with weights and a well defined architecture. The main characteristics of the NNQA are that inputs and weights are represented by a string of qubits, and network follows a unitary evolution. Based on these characteristics  we will propose the quantum perceptron over a field. 

\subsection{Neuron operations}

 We propose a quantum perceptron with the following properties: it can be trained with a quantum or classical algorithm, we can put all neural networks for a given architecture in superposition, and if the weights are in the computational basis the quantum perceptron acts like the classical perceptron. One of the difficulties to define a quantum perceptron is that the set of n-dimensional qubits, sum and (tensor) product operators do not form a field (as shown in Proposition~\ref{proposition:sumnotfield}). Therefore, the first step in the definition of the quantum perceptron is to define a new appropriate sum and multiplication of qubits.

\begin{proposition}
\label{proposition:sumnotfield}
The set of qubits under sum + of qubits and tensor product $\otimes$ is not a field.
\end{proposition}
\begin{proof}
The null vector has norm 0 and is not a valid qubit. Under + operator the null vector is unique, so there is not a null vector in the set of qubits. Then we cannot use + operator to define a field in the set of qubits. 

Tensor product between two qubits results in a compound system with two qubits. So the space of qubits is not closed under the $\otimes$ operator and we cannot use $\otimes$ operator to define a field in the set of qubits. 
\end{proof}

 We will define unitary operators to perform sum $\oplus$ and product $\odot$ of qubits based in the field operations. Then we will use these new operators to define a quantum neuron. Let  $F(\oplus,\odot)$ be a finite field. We can associate the values $a\in F$ to vectors (or qubits) $\ket{a}$ in a basis of a vector space  $\field{V}$. If $F$ has $n$ elements the vector space will have dimension $n$.

 Product operation  $\odot$ in $F$ can be used to define a new product between vectors in $\field{V}$. Let $\ket{a}$ and $\ket{b}$ be qubits associated with  scalars 
$a$ and $b$, we define $\ket{a}\odot\ket{b}=\ket{a \odot b}$   such that 
$\ket{a \odot b}$  is related with the scalar  $a \odot b$. 
 This product send basis elements to basis elements, so we can define a linear 
operator $\qop{P}:\field{V}^3 \rightarrow \field{V}^3$ as in Equation~\eqref{eq:odot}.  
 
\begin{equation}
\qop{P}\ket{a}\ket{b}\ket{c} = \ket{a}\ket{b}\ket{c \oplus \left(a \odot 
b\right)}
\label{eq:odot}
\end{equation}
We show in Proposition~\ref{proposition:unitariedade} that the $\qop{P}$ operator is 
unitary, therefore $\qop{P}$ is a valid quantum operator.

\begin{proposition}
\label{proposition:unitariedade}
$\qop{P}$ is a unitary operator.
\end{proposition}
\begin{proof}

Let $B=\{\ket{a_1}, \ket{a_2}, \cdots, \ket{a_n}\}$ be a computational basis of a vector space $\vs{V}$, where we associate $\ket{a_i}$ with the element $a_i$ of the 
finite field $F(\oplus,\odot)$.  The set $B^3 = \{\ket{a_i}\ket{a_j}\ket{a_k}\}$ 
with $1 \leq i, j, k \leq n$ is a computational basis of vector space $\vs{V}^3$.
We will show that $\qop{P}$ sends elements of basis $B^3$ in distinct elements 
of basis $B^3$.

$\qop{P}$ sends vectors in base $B^3$ to vectors in base $B^3$. Let $a$, $b$, 
$c$ in $B$. 
$\qop{P} \ket{a} \ket{b} \ket{c} = \ket{a} \ket{b} \ket{c \oplus (a \odot b)} 
$. Operators $\oplus$ and $\odot$ are well defined, then $\ket{c \oplus (a \odot 
b)} \in B$ and $\ket{a} \ket{b} \ket{c \oplus (a \odot b)} \in B^3$.

$\qop{P}$ is injective. Let $a, b, c, a_1, b_1, c_1 \in F$ such that $\qop{P} 
\ket{a}\ket{b} \ket{c} = \qop{P} \ket{a_1}\ket{b_1} \ket{c_1}$ . By the 
definition of operator $\qop{P}$ $a = a_1$ and $b = b_1$ . When we apply the 
operator we get $\ket{a}\ket{b}\ket{c \oplus (a \odot b)} =
\ket{a}\ket{b}\ket{c_1 \oplus (a \odot b)}$. Then $c_1 \oplus (a \odot b) = c 
\oplus (a \odot b)$ . By the field properties we get $c_1 = c$.

 An operator over a vector space is unitary if and only if the operator sends some orthonormal basis to some orthonormal basis~\cite{hoffman:71}. As $\qop{P}$ 
is injective and sends vectors in base $B^3$  to vectors in base $B^3$, we 
conclude that $\qop{P}$ is an unitary operator.
\end{proof}

 With a similar construction, we can use the sum operation $\oplus$ in $F$ to define a unitary sum operator $\qop{S}:\vs{V}^3 \rightarrow \vs{V}^3$. Let $\ket{a}$ and 
$\ket{b}$ be qubits associated with  scalars $a$ and $b$, we define $\ket{a} \oplus \ket{b}=\ket{a \oplus b}$   such that 
$\ket{a \oplus b}$  is related with the scalar  $a \oplus b$. We define the 
unitary quantum operator $S$ in Equation~\eqref{eq:oplus}.

\begin{equation}
\qop{S}\ket{a}\ket{b}\ket{c} = \ket{a}\ket{b}\ket{c \oplus (a \oplus b)}
\label{eq:oplus}
\end{equation}

 We denote product and sum of vectors over $F$ by $\ket{a}\odot 
 \ket{b}=\ket{a \odot b}$  and $\ket{a} \oplus \ket{b} = 
\ket{a \oplus b}$ to represent, respectively, 
$\qop{P}\ket{a}\ket{b}\ket{0}=\ket{a}\ket{b}\ket{a\odot b}$ and 
$\qop{S}\ket{a}\ket{b}\ket{0}=\ket{a}\ket{b}\ket{a \oplus 
b}$. 

\subsection{Quantum perceptron over a field}

 Using $\qop{S}$ and $\qop{P}$ operators we can define a quantum perceptron 
analogously as the classical one. Inputs $\ket{x_i}$, weights $\ket{w_i}$ and output $\ket{y}$ will be unit vectors (or qubits) in 
$\vs{V}$ representing scalars in a field $F$. Equation~\eqref{eq:qpf} describes the 
proposed quantum perceptron.
\begin{equation}
\ket{y} = \bigoplus_{i=1}^{n}\ket{x_i} \odot \ket{w_i}
\label{eq:qpf}
\end{equation}
 If the field $F$ is the set of rational numbers,  then Definition~\ref{def:neuron} without activation function $f$ correspond to  Definition~\ref{eq:qpf} when inputs, and weights are in the computational basis.
 
 The definition in Equation~\eqref{eq:qpf} hides several ancillary qubits. The 
complete configuration of a quantum perceptron is given by the state $\ket{\psi}$ described in Equation~\eqref{eq:qpfanc}, 
\begin{equation}
\begin{split}
\ket{\psi} = \left|x_1,\cdots, x_n,w_1,\cdots, w_n,\right.\\ \left.p_1,\cdots, 
p_n,s_2,\cdots, s_{n-1},y \right\rangle
\end{split}
\label{eq:qpfanc}
\end{equation}
where 
$\ket{x}= \ket{x_1,\cdots, x_n}$ is the input quantum register, 
$\ket{w}=\ket{w_1,\cdots, w_n}$ is the weight quantum register, 
$\ket{p}=\ket{p_1,\cdots, p_n}$ is an ancillary quantum register used to store 
the products $\ket{x_i \odot w_i}$, $\ket{s}=\ket{s_2,\cdots, s_{n-1}}$ is an 
ancillary quantum register used to store sums, and 
$\ket{y}$ is the output quantum register. From 
Equation~\eqref{eq:qpfanc} one can see that to put several or all possible neural networks  in superposition one can simply put the weight quantum register in superposition. Then a single quantum neuron can be in a superposition of several neurons simultaneously.

  A quantum perceptron over a finite $d$-dimensional field and with $n$ inputs 
needs $(4n-1) \cdot d$ quantum bits to perform its computation. There are $n$ quantum registers to store inputs $x_i$, $n$ quantum registers to store weights $w_i$, $n$ quantum 
registers $p_i$ to store the products $w_i \odot x_i$, $n-2$ quantum registers 
$\ket{s_i}$ to store sums $\sum_{k=1}^i p_i$ and one quantum register to store the output $\ket{y} = \sum_{k=1}^n p_i$.

 We show now a neuron with 2 inputs to illustrate the workings of the quantum neuron. Suppose  $F = \field{Z}_2   = \{0,1\}$. As the field has only two elements we need only two orthonormal quantum states to represent the scalars. We choose the canonical ones $0 \leftrightarrow \ket{0}$ and $1 \leftrightarrow 
\ket{1}$.

 Now we define sum $\oplus$ and multiplication $\odot$ operators based on the sum and multiplication in $\field{Z}_2$. The operators $\qop{S}$ and $\qop{P}$ 
are shown, respectively, in Equations~\eqref{eq:sum} and~\eqref{eq:prod}. 

\begin{equation}
S=
\begin{bmatrix}
1 & 0 & 0 & 0 & 0 & 0 & 0 & 0\\
0 & 1 & 0 & 0 & 0 & 0 & 0 & 0\\
0 & 0 & 0 & 1 & 0 & 0 & 0 & 0\\
0 & 0 & 1 & 0 & 0 & 0 & 0 & 0\\
0 & 0 & 0 & 0 & 0 & 1 & 0 & 0\\
0 & 0 & 0 & 0 & 1 & 0 & 0 & 0\\
0 & 0 & 0 & 0 & 0 & 0 & 1 & 0\\
0 & 0 & 0 & 0 & 0 & 0 & 0 & 1\\
\end{bmatrix}
\label{eq:sum}
\end{equation}

\begin{equation}
P=
\begin{bmatrix}
1 & 0 & 0 & 0 & 0 & 0 & 0 & 0\\
0 & 1 & 0 & 0 & 0 & 0 & 0 & 0\\
0 & 0 & 1 & 0 & 0 & 0 & 0 & 0\\
0 & 0 & 0 & 1 & 0 & 0 & 0 & 0\\
0 & 0 & 0 & 0 & 1 & 0 & 0 & 0\\
0 & 0 & 0 & 0 & 0 & 1 & 0 & 0\\
0 & 0 & 0 & 0 & 0 & 0 & 0 & 1\\
0 & 0 & 0 & 0 & 0 & 0 & 1 & 0\\
\end{bmatrix}
\label{eq:prod}
\end{equation}

Using $\qop{S}$ and $\qop{P}$ operators we describe the quantum neuron  
$\qop{N}$ in Equation~\eqref{eq:neuronexample}. The subscripts in operators indicate the qubits upon which they will be applied. 
\begin{equation}
\label{eq:neuronexample}
 \qop{N} = \qop{S}_{p_1p_2,y}\qop{P}_{x_2w_2,p_2}\qop{P}_{x_1w_1,p_1}
\end{equation}

 For our execution example, we define $\ket{ x_1 x_2 }=\ket{01}$,  $\ket{w_1 
}=\frac{1}{\sqrt{2}}(\ket{0}+\ket{1})$  and 
$\ket{w_2}=\frac{1}{\sqrt{2}}(\ket{0}+\ket{1})$. The initial configuration of 
the quantum perceptron is  $\ket{\psi_0}$ described in Equation~\eqref{eq:initialconf}. The initial configuration has all possible weights in the set $\{0,1\}^2$ and applying the QP  will result the output for each weight simultaneously. 

\begin{equation}
\begin{split}
\ket{x_1x_2}\ket{w_1w_2}\ket{p_1p_2}\ket{s}\ket{y}= \\
\frac{1}{2}\ket{01}\left( \ket{00} + \ket{01}+\ket{11}+\ket{11}   
\right)\ket{00}\ket{0}\ket{0} =  \\\frac{1}{2}\left( 
\ket{01}\ket{00}\ket{00}\ket{0}\ket{0} + 
\ket{01}\ket{01}\ket{00}\ket{0}\ket{0} + \right.\\
\ket{01}\ket{10}\ket{00}\ket{0}\ket{0} +
\left. \ket{01}\ket{11}\ket{00}\ket{0}\ket{0}
\right)
\end{split}
\label{eq:initialconf}
\end{equation}

  The action of the quantum perceptron $\qop{N}$ over $\ket{\psi_0}$ is shown in 
Equation~\eqref{eq:exrun}, where $\qop{N}$ calculates the output for all possible weights (or all neurons in superposition) in only one single run. 
\begin{equation}
\begin{split}
\qop{N}\ket{\psi_0} =  \frac{1}{2}\qop{N}( 
\ket{01}\ket{00}\ket{00}\ket{0}\ket{0} + \\
\ket{01}\ket{01}\ket{00}\ket{0}\ket{0} +  
\ket{01}\ket{10}\ket{00}\ket{0}\ket{0} + \\
\ket{01}\ket{11}\ket{00}\ket{0}\ket{0})= \\
\frac{1}{2}( 
\qop{N}\ket{01}\ket{00}\ket{00}\ket{0}\ket{0} + 
\qop{N}\ket{01}\ket{01}\ket{00}\ket{0}\ket{0} +  \\
\qop{N}\ket{01}\ket{10}\ket{00}\ket{0}\ket{0} + 
\qop{N}\ket{01}\ket{11}\ket{00}\ket{0}\ket{0})= \\
\frac{1}{2}( 
\qop{S}_{p_1p_2y}\ket{01}\ket{00}\ket{0\odot 0,1 \odot 0}\ket{0}\ket{0} + \\
\qop{S}_{p_1p_2y}\ket{01}\ket{01}\ket{0\odot 0,1 \odot 1}\ket{0}\ket{0} +  \\
\qop{S}_{p_1p_2y}\ket{01}\ket{10}\ket{0\odot 1,1 \odot 0}\ket{0}\ket{0} + \\
\qop{S}_{p_1p_2y}\ket{01}\ket{11}\ket{0\odot 1,1 \odot 1}\ket{0}\ket{0})= \\
\ket{01}\ket{00}\ket{0,0}\ket{0\oplus 0}\ket{0} + \\
\ket{01}\ket{01}\ket{0,1}\ket{0\oplus 1}\ket{0} +  \\
\ket{01}\ket{10}\ket{0,0}\ket{0\oplus 0}\ket{0} + \\
\ket{01}\ket{11}\ket{0,1}\ket{0\oplus 1}\ket{0})= \\
\ket{01}\ket{00}\ket{0,0}\ket{0}\ket{0} + 
\ket{01}\ket{01}\ket{0,1}\ket{1}\ket{1} +  \\
\ket{01}\ket{10}\ket{0,0}\ket{0}\ket{0} + 
\ket{01}\ket{11}\ket{0,1}\ket{1}\ket{1})
\end{split}
\label{eq:exrun}
\end{equation}

\subsubsection{Neural network architecture}
We start this Subsection with an example of a classical multilayer perceptron and show an equivalent representation in a quantum computer. Let $N$ be the neural network described in Fig.~\ref{fig:architecture}. 
\begin{figure}
\begin{center}
\begin{tikzpicture}[->]
 \tikzstyle{input} = [fill=black, draw]
 \tikzstyle{neuron}=[circle,draw,minimum size=20,inner sep=0]
 
 \foreach \name / \y in {0,3,6}
	{
	\node at (-5.4,-\y){\pgfmathparse{int(\y/3+1)}$x_{\pgfmathresult}$};
	\node[input] (In-\name) at (-5,-\y) {};
	}

 \foreach \name / \y in {0,3,6}
	{
	\node[neuron] (I-\name) at (-2,-\y) {};
	\node at (-2,-\y) [] {$\sum$};
	}
 \foreach \name / \y in {2,4}
	{
	\node[neuron] (H-\name) at (0,-\y) {};
	\node at (0,-\y) [] {$\sum$};
	}

            
 \path(In-0) edge node[above, sloped,pos =0.2]{$w_{11}$} (I-0);
 \path(In-3) edge node[above, sloped,pos =0.15]{$w_{12}$} (I-0);
 \path(In-6) edge node[above, sloped,pos =0.10]{$w_{13}$} (I-0);
 
 \path(In-0) edge node[above, sloped,pos =0.2]{$w_{21}$} (I-3);
 \path(In-3) edge node[above, sloped,pos =0.2]{$w_{22}$} (I-3);
 \path(In-6) edge node[below, sloped,pos =0.2]{$w_{23}$} (I-3);
 
 \path(In-0) edge node[below, sloped,pos =0.15]{$w_{31}$} (I-6);
 \path(In-3) edge node[below, sloped,pos =0.1]{$w_{32}$} (I-6);
 \path(In-6) edge node[below, sloped,pos =0.2]{$w_{33}$} (I-6);

 
 \path(I-0) edge node[above, sloped,pos =0.2]{$w_{14}$} (H-2);
 \path(I-3) edge node[above, sloped,pos =0.2]{$w_{24}$} (H-2);
 \path(I-6) edge node[above, sloped,pos =0.2]{$w_{34}$} (H-2);
 
 \path(I-0) edge node[below, sloped,pos =0.2]{$w_{15}$} (H-4);
 \path(I-3) edge node[below, sloped,pos =0.2]{$w_{25}$} (H-4);
 \path(I-6) edge node[below, sloped,pos =0.2]{$w_{35}$} (H-4);
            
\foreach \name / \y in {2,4}
	{
	
	\path (H-\name) edge node[above] 
{\pgfmathparse{int(\y/2)}$y_\pgfmathresult$}(0.75,-\y);
	}
\end{tikzpicture}	
\end{center}
\caption{Architecture of the multilayer quantum perceptron over a field.}
\label{fig:architecture}
\end{figure}
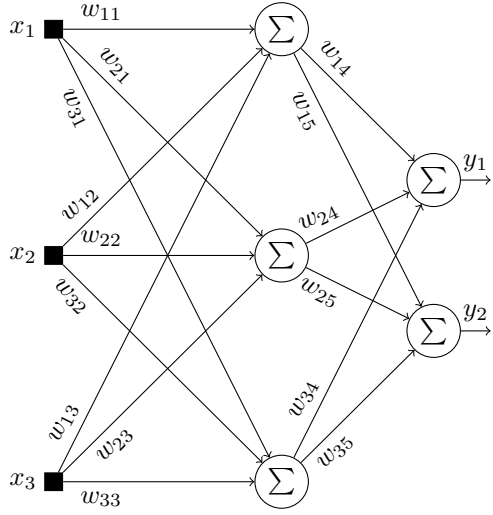
The output  of this network can be calculated as  $y=L_2 \cdot L_1 \cdot x$ using the three matrices $L_1$, $L_2$ and $x$  described in Equation~\eqref{eq:architecturematrix}.

\begin{equation}
\label{eq:architecturematrix}
\begin{split}
L_1 =
 \begin{bmatrix}
  w_{11} & w_{12} & w_{13} \\
  w_{21} & w_{22} & w_{23} \\
  w_{31} & w_{32} & w_{33}
 \end{bmatrix}, \\
L_2 = 
\begin{bmatrix}
w_{14} & w_{24} & w_{34} \\
w_{15} & w_{25} & w_{35} \\
\end{bmatrix},\ \
x=\begin{bmatrix}
   x_1 \\ x_2 \\ x_3
  \end{bmatrix}
  \end{split}
\end{equation}

Weights, inputs and outputs in a classical neural network are real numbers. Here we suppose finite memory and we use elements of a finite field $(F,\oplus, 
\odot)$  to represent the neural network parameters. We can define a quantum operator $M_{3\times 3, 3 \times 1}$ that multiplies a $3 \times 3$ matrix with 
a $3 \times 1 $ matrix. If $L_1 \cdot x = \begin{bmatrix}o_1 &  o_2 & 
o_3\end{bmatrix}^t$ we define the action of $M_{3\times 3, 3 \times 1}$ in 
Equation~\eqref{eq:qmatrix}, where $w_i = w_{i1}, w_{i2}, w_{i3}$.
\begin{equation}
\label{eq:qmatrix}
\begin{split}
 M_{3\times 3, 3 \times 1}\ket{w_1,w_2,w_3, x_1,x_2,x_3,0,0,0} = \\
 \ket{w_1,w_2,w_3, x_1,x_2,x_3,o_1,o_2,o_3}
 \end{split}
\end{equation}

Each layer of the quantum perceptron over a field can be represented by an arbitrary matrix as in Equation~\eqref{eq:qmlarch}, 
\begin{equation}
\label{eq:qmlarch}
 M_{2\times 3, 3 \times 1} M_{3\times 3, 3 \times 
1}\ket{L_2}\ket{L_1}\ket{x}\ket{000}\ket{00}
\end{equation}
where $M_{3\times 3, 3  \times 1}$ acts on $\ket{L_1}$, $\ket{x}$ with output in register initialized with $\ket{000}$; and $M_{2\times 3, 3 \times 1}$ acts on $\ket{L_2}$, the output of the first operation, and the last quantum register.
This matrix approach can be used to represent any feed-forward multilayer 
quantum perceptron over a field  with any number of layers.

We suppose here that the training set and desired output are composed of classical data and that the data run forward. The supposition of classical desired output will allow us to superpose neural network configurations with its performance, as we will see in the next section.

\subsubsection{Learning algorithm}
 In this Subsection, we present a learning algorithm that effectively uses quantum superposition to train a quantum perceptron over a field. Algorithms based on superposition have been proposed previously in~\cite{panella:09,panella2009neurofuzzy,zhou:12}. In these papers, a non-linear quantum operator proposed in~\cite{PhysRevLett.81.3992}, is used in the learning process.  In~\cite{panella:09} performances of neural networks in superposition are entangled with its representation. A non-linear algorithm is used to recover a neural network configuration with performance greater than a given threshold $\theta$. A non-linear 
algorithm is used to recover the best neural network configuration. In~\cite{panella2009neurofuzzy} the nonlinear quantum operator is used in the learning process of a neurofuzzy network. In~\cite{zhou:12} a quantum associative neural network is proposed where a non-linear quantum circuit is used to increase the pattern recalling speed.

  We propose a variant of the learning algorithm proposed in~\cite{panella:09}. The proposed quantum algorithm is named Superposition based Architecture Learning (SAL) algorithm.  In the SAL algorithm the superposition of neural networks will store its performance entangled with its representation, as in~\cite{panella:09}. Later we will use a non-linear quantum operator to recover the architecture and weights of the neural network configuration with best performance. 

  In the classical computing paradigm, the idea of presenting an input pattern to all possible neural networks architectures is impracticable. To perform this idea classically one will need to create several copies of the neural network (one for each configuration and architecture) to receive all the inputs and compute in parallel the corresponding outputs. After calculating the output of each pattern for each neural network configuration, one can search the neural configuration with best performance. Yet classically the idea of SAL learning is presented in
Fig.~\ref{fig:framework}. For some neural network architectures, all the patterns in the training set  $P = \{p_1,p_2,\cdots,p_k\}$ are presented to each of the neural network configurations. Outputs are calculated and then one can search the best neural network parameters.

\begin{figure}
\begin{center}
\resizebox{0.5\columnwidth}{!}{
\begin{tikzpicture}

\draw  (10,10) rectangle (12,11);
\node at (11,10.5) {QPF1.1};
\node at (11,9.5) {.};
\node at (11,9) {.};
\node at (11,8.5) {.};
\draw  (10,8) rectangle (12,7);
\node at (11,7.5) {QPF1.$n_1$};
\node at (11,6.5) {Architecture 1};
\draw [densely dashed] (9.5,11.5) rectangle (12.5,6);

\node at (11,5.5) {.};
\node at (11,5.0) {.};
\node at (11,4.5) {.};

\draw  (10,2.5) rectangle (12,3.5);
\node at (11,3) {QPF1.1};
\node at (11,2) {.};
\node at (11,1.5) {.};
\node at (11,1) {.};
\draw  (10,0.5) rectangle (12,-0.5);
\node at (11,0) {QPF1.$n_2$};
\node at (11,-1) {Architecture $m$};
\draw [densely dashed] (9.5,4) rectangle (12.5,-1.5);
\draw  (5,5.5) rectangle (7,4.5);
\node at (6,5) {input};
\node (v1) at (7,5) {};
\node (v2) at (8,5) {};
\node (v5) at (8,7.5) {};
\node (v3) at (8,10.5) {};
\node (v4) at (10,10.5) {};
\node (v6) at (10,7.5) {};
\node (v9) at (8,3) {};
\node (v7) at (8,0) {};
\node (v8) at (10,0) {};
\node (v10) at (10,3) {};
\draw  (v1) edge (v2);
\draw  (v2) edge (v3);
\draw  (v3) edge[->] (v4);
\draw  (v5) edge [->](v6);
\draw  (v2) edge (v7);
\draw  (v7) edge[->] (v8);
\draw  (v9) edge[->] (v10);
\node (v11) at (12,10.5) {};
\node (v12) at (13.5,10.5) {};
\node (v13) at (12,7.5) {};
\node (v14) at (13.5,7.5) {};
\node (v15) at (12,3) {};
\node (v16) at (13.5,3) {};
\node (v17) at (12,0) {};
\node (v18) at (13.5,0) {};
\draw  (v11) edge[->] (v12);
\draw  (v13) edge[->] (v14);
\draw  (v15) edge[->] (v16);
\draw  (v17) edge[->] (v18);
\draw  (5,-1) rectangle (7,-2.5);
\node at (6,-1.5) {desired};
\node at (6,-2) {output};
\draw  (13.5,11) rectangle (17,10);
\node at (15.5,10.5) {Performance 1.1};
\draw  (13.5,8) rectangle (17,7);
\node at (15.5,7.5) {Performance $m$.$n_1$};
\draw  (13.5,3.5) rectangle (17,2.5);
\node at (15.5,3) {Performance 1.$n_1$};
\draw  (13.5,0.5) rectangle (17,-0.5);
\node at (15.5,0) {Performance $m$.$n_2$};
\node (v19) at (7,-2) {};
\node (v20) at (14,-2) {};
\node (v21) at (14,11.5) {};
\draw [dashed] (v19) edge (v20);
\draw [dashed] (v20) edge[->] (v21);
\end{tikzpicture}
}
\end{center}
\caption{Superposition based framework}
\label{fig:framework}
\end{figure}
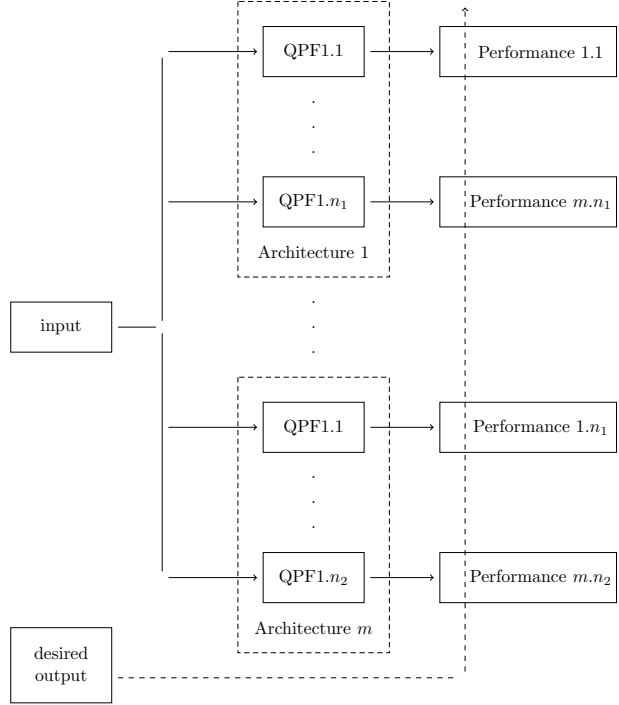

 In a quantum computer, the framework described in Fig.~\ref{fig:framework} 
can be implemented  without the hardware limitations showed in the classical implementation. Let $\qop{N}_0$, $\qop{N}_1$, $\cdots$, $\qop{N}_{m-1}$ be $m$ quantum operators representing neural networks with different architectures. A quantum circuit with quantum registers architecture selector $\ket{a}$ with $\lceil\log_2(m)\rceil$ qubits, input $\ket{x}$, weight $\ket{w}$ and output $\ket{o}$ can be created, where 
operator $\qop{N}_i$ is applied to $\ket{x, w, o}$ if and only if $\ket{a} = \ket{i}$. In Fig.~\ref{fig:quantumarch} this approach is illustrated with $m=4$.

\begin{figure}
 \[  \Qcircuit @C=0.22em @R=1.0em {
\lstick{\ket{a_1}} & \qw & \ctrlo{1} &  \qw & \qw &  \qw &  \qw &  \qw & 
\ctrlo{1} &  \qw & \qw & \qw &  \qw &  \qw & \ctrl{1} &  \qw &  \qw & \qw &  
\qw 
&  \qw & \ctrl{1} &  \qw \\
\lstick{\ket{a_2}} & \qw & \ctrlo{1} &  \qw & \qw &  \qw &  \qw &  \qw & 
\ctrl{1} &  \qw & \qw & \qw &  \qw &  \qw & \ctrlo{1} &  \qw & \qw & \qw &  \qw 
&  \qw & \ctrl{1} &  \qw \\
\lstick{\ket{x}}&\qw  & \multigate{1}{\qop{N}_0} & \qw & \qw & \qw & \qw & \qw 
& \multigate{1}{\qop{N}_1} & \qw & \qw &\qw & \qw & \qw & 
\multigate{1}{\qop{N}_2} & \qw 
&\qw & \qw & \qw & \qw & \multigate{1}{\qop{N}_3} &   \qw   \\
&\lstick{\ket{w}}   & \ghost{\qop{N}_0} &  \qw & \qw & \qw& \qw & \qw & 
\ghost{\qop{N}_1} &  \qw & \qw& \qw& \qw& \qw & \ghost{\qop{N}_2} &  \qw &\qw 
&\qw &\qw & 
\qw & \ghost{\qop{N}_3} & \qw \\
 }\]
 
 \caption{Circuit to create a superposition with four neural network 
architectures.}
\label{fig:quantumarch}
\end{figure}
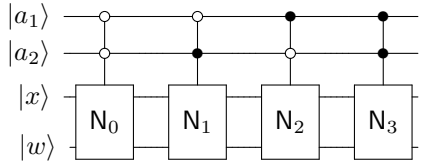

If the qubits in quantum registers $\ket{a}$ and $\ket{w}$ are initialized with 
the quantum state $\qop{H}\ket{0}$, the circuit will be in a superposition 
representing all possible weights configuration for each architecture. 
Initializing the quantum register $\ket{x}$ with a pattern $p$, it is possible 
to present the pattern $p$ to all neural network configurations in the 
superposition simultaneously.

Algorithm~\ref{alg:lssa} is the SAL algorithm. SAL is a quantum-learning algorithm for any quantum neural network model in which input $\ket{p}$, output $\ket{o}$, weights $\ket{w}$, architecture selectors $\ket{a}$ and desired output $\ket{d}$ are represented in separated quantum registers. The main idea of SAL algorithm is to create a superposition of all possible neural network configurations in a finite set of architectures and apply a non-linear quantum operator to recover the architecture and weights of a neural network configuration with a desired performance.

The Algorithm~\ref{alg:lssa} initialization is performed in Steps~\ref{line:1} to~\ref{line:4}. Step~\ref{line:1} defines $m$ quantum operators representing multi-layers QPF with different architectures. Steps~\ref{line:qcircuit} and~\ref{line:2} initialize all the weights and architecture selectors with the quantum state $\qop{H}\ket{0}$. After this step we have all possible neural network configurations for the given architectures in superposition. In Steps~\ref{line:3} and~\ref{line:4} quantum registers performance and objective are initialized respectively, with the quantum states $\ket{0}_n$ and $\ket{0}$.

  The for loop starting in Step~\ref{line:5} is repeated for each pattern $p$ of the data set. Step~\ref{line:6} initializes quantum registers $p$, $o$ and $d$ respectively with a pattern $\ket{p}$, state $\ket{0}$ and its desired output $\ket{d}$. Step~\ref{line:7} presents the pattern  to the neural networks, and the outputs are calculated. In this moment, the pattern is present to all neural networks configurations, because the weights and architecture selectors quantum registers are in a superposition of all possible weights and architectures. In Steps~\ref{line:8} to~\ref{line:10}, it is verified for each configuration in the superposition if the desired output $\ket{d}$ is equal to the calculated output $\ket{o}$. If they match, is added the value 1 for the performance quantum register. Step~\ref{line:11} is performed to allow the initialization of the next for loop.

\begin{equation}
\begin{split}
\ket{a}\ket{w}\ket{\mathrm{\emph{performance}}}\ket{objective} =\\ 
\sum_{\substack{w \in W, \\ a \in A} 
}\ket{a}\ket{w}\ket{\mathrm{\emph{performance}}(w)}\ket{0} 
\end{split}
\label{eq:afterfor}
\end{equation}
After the execution of the for loop, the state of quantum registers 
weights $w$, architecture selectors $a$,
$\mathrm{\emph{performance}}$ and $objective$ can be described as in 
Equation~\eqref{eq:afterfor}, where $A$ is the set of architectures and $W$ is the set of all possible weights.

 \begin{algorithm}[H]
 \begin{algorithmic}[1]
  \STATE Let $\qop{N}_0$, $\qop{N}_1$, $\cdots$, $\qop{N}_m$ be quantum 
operators representing multi-layers QPF with different architectures. 
\label{line:1}
 \STATE Create a quantum circuit where the i-th network acts if and only if 
$\ket{a}=\ket{i}$.\label{line:qcircuit}\\
\STATE	Initialize all weights quantum registers with the quantum state 
$\qop{H}\ket{0}$. 
\label{line:2}\\
\STATE Initialize all architecture quantum registers with quantum 
state $\qop{H}\ket{0}$. \\
\STATE	Initialize a quantum register $\ket{\mathrm{\emph{performance}}}$ with 
the state 
$\ket{0}_n$. \label{line:3}\\
\STATE	Initialize a quantum register $\ket{objective}$ with the state 
$\ket{0}$.\label{line:4}\\
\FOR{ \textbf{each} pattern $p$ and desired output $d$ in the training set }
\label{line:5}
		\STATE Initialize the register $p$, $o$ , $d$ with the quantum 	
state $\ket{p,0,d}$. \label{line:6}\\ 
		\STATE Calculate $\qop{N}\ket{p}$ to calculate network output 
in register $\ket{o}$. \label{line:7}\\
		\IF{ $\ket{o}=\ket{d}$ \label{line:8}} 
			\STATE Add 1 to the register 
$\ket{\mathrm{\emph{performance}}}$\label{line:9}
		\ENDIF
		\label{line:10}
		\STATE Calculate $\qop{N}^{-1}\ket{p}$ to restore 
$\ket{o}$.\label{line:11}

	\ENDFOR
	\STATE Perform a non-linear quantum search  to 
recover a neural network configuration and architecture with desired 
performance. 
\label{line:14}
\end{algorithmic}
\caption{SAL}
\label{alg:lssa}
\end{algorithm}
 
Steps~\ref{line:1} to~14 of Algorithm~\ref{alg:lssa} can be performed using only linear quantum operators. In Step~\ref{line:14} a non-linear quantum operator \qop{NQ} proposed in~\cite{PhysRevLett.81.3992} will be used. Action of \qop{NQ} is described in Equation~\eqref{eq:nonlin} if at least one $\ket{c_i}$ is equal to $\ket{1}$ otherwise its action is described in Equation~\eqref{eq:nonlin2}.

\begin{equation}
\qop{NQ} \left(\sum_i \ket{\psi_i}\ket{c_i} \right)= \left(\sum_i \ket{\psi_i}\right)\ket{1}
\label{eq:nonlin}
\end{equation}
\begin{equation}
\qop{NQ} \left(\sum_i \ket{\psi_i}\ket{c_i}\right) = \left(\sum_i \ket{\psi_i}\right)\ket{0}
\label{eq:nonlin2}
\end{equation}

The non-linear search used in Step~\ref{line:14} is described in Algorithm~\ref{alg:nonlinearsearch}. The for loop in Step 1 of Algorithm~\ref{alg:nonlinearsearch} indicates that the actions need to be repeated for each quantum bit in the architecture and weights quantum registers. Steps 3 to 5 set the objective quantum register $\ket{o}$ to $\ket{1}$ if the performance quantum register $p$ is greater than a given threshold $\theta$. After this operation the state of quantum registers $a$, $w$ and $o$ can be described as in 
Equation~\eqref{eq:objective}.
\begin{equation}
\sum_{\mathclap{\substack{w\in \left(P(a,w) < \theta\right), \\ \ket{b} \neq 
\ket{i} }} }\ket{a}\ket{w}\ket{0} +  \sum_{\mathclap{\substack{w\in 
\left(P(a,w) \geq \theta\right),\\ \ket{b} = \ket{i} }}}\ket{a}\ket{w}\ket{1} 
\label{eq:objective}
\end{equation}
Now that quantum register objective is set to 1 in the desired configurations, it is possible to perform a quantum search to increase the probability amplitude of the best configurations. 

\begin{algorithm}[H]
 \begin{algorithmic}[1]
  \FOR{\textbf{each} quantum bit $\ket{b}$ in quantum registers 
$\ket{a}\ket{w}$}
\FOR{$i=0$ \TO $1$}
   \IF{$\ket{b} = \ket{i}$ and $\ket{p} > \theta$}
  \STATE Set $\ket{o}$ to $\ket{1}$
  \ENDIF
  \STATE Apply $\qop{NQ}$ to $\ket{o}$

  \IF{$\ket{o} = \ket{1}$}
  \STATE Apply $\qop{X}^i \cdot \qop{NQ}$ to qubit 
$\ket{b}$
  \STATE Apply $\qop{X}$ to $\ket{o}$
  \ENDIF
  \ENDFOR
  \ENDFOR
 \end{algorithmic}
\caption{Non-linear quantum search}
\label{alg:nonlinearsearch}
\end{algorithm}

Step 6 applies $\qop{NQ}$ to quantum register $\ket{o}$. If there is at least one configuration 	with $\ket{b} = \ket{i}$ then the action of $\qop{NQ}$ will set $\ket{o}$ to $\ket{1}$. In this case, Steps 7 to 10  set qubit $\ket{b}$ from a superposed state to the computational basis state $\ket{i}$.

Algorithm~\ref{alg:nonlinearsearch} performs an exhaustive non-linear quantum search in the architecture and weights space. If there is a neural network configuration with the desired performance in initial superposition, the search will return one of these configurations. Otherwise the algorithm does not change the initial superposition and the procedure can be repeated with another performance threshold.

The computational cost of Steps 1 and 2 of SAL is $O(m)$ and depends on the number of neural networks architectures. Steps 3 to 6 has computational cost $O(m+n_w)$, where $n_w$ is the number of qubits to represent the weights. The for starting in Step 7 will be executed $p$ times and each inner line has constant cost. Step 15 is detailed in Algorithm~\ref{alg:nonlinearsearch}.  Steps 3 to 9 of Algorithm~\ref{alg:nonlinearsearch} have constant computational cost and it will be repeated $2\cdot \left( m + n_w \right)$ times. The overall cost of the SAL algorithm is $O(p+m+n_w)$ where $p$ is the number of patterns in the training set.

\section{Discussion}
\label{sec:discussion}

Classical neural networks have limitations, such as  i) the lack of an algorithm to determine optimal architectures, ii) memory capacity and iii) high cost learning algorithms. In this paper, we investigate how to use quantum computing to deal with  limitation iii). To achieve this objective, we define a quantum neural network model named quantum perceptron over a field QPF and a nonlinear quantum learning algorithm that performs an exhaustive search in the space of weights and architectures.

We have shown that previous models of quantum perceptron cannot be viewed as a direct quantization of the classical perceptron. In other models of quantum neural networks weights and inputs are represented by a string of qubits, but  the set of all possible inputs and weights with inner neuron operations does not form a field and there is no guarantee that they are well defined operations. To define QPF we propose quantum operators to perform addition and multiplication such that the qubits in a computational basis form a field with these operations. QPF is the unique neuron with these properties. We claim that QPF can be viewed as a direct quantization of a classical perceptron, since when the qubits are in the computational basis the QPF acts exactly as a classical perceptron. In this way, theoretical results obtained for the classical perceptron remains valid to QPF.

There is a lack of learning algorithms to find the optimal architecture of a neural network to solve a given problem. Methods for searching  near optimal architecture use heuristics and perform  local search in the space of architectures as \emph{eg.} trough evolutionary algorithms or meta-learning.   We propose an algorithm that solves this open problem using a nonlinear quantum search algorithm based on the learning algorithm of the NNQA. The proposed learning algorithm, named SAL, performs a non-linear exhaustive search in the space of architecture and weights and finds the best architecture (in a set of previously defined architectures) in linear time in relation to the number of patterns in the training set. SAL uses quantum superposition to allow initialization of all possible architectures and weights in a way that the architecture search is not biased by a choice of weights configuration. The desired architecture and weight configuration is obtained by the application of the nonlinear search algorithm and we can use the obtained neural network as a classical neural network. The QPF and SAL algorithm extend our theoretical knowledge in learning in quantum neural networks.

Quantum computing is still a theoretical possibility with no actual computer, an empirical evaluation of the QPF in  real world problems is not yet possible, ``quantum computing is far from actual application~\cite{panella:09}''. Studies necessary to investigate the generalization capabilities of SAL algorithm through a cross-validation procedure cannot be accomplished with actual technology. 

The simulation of the learning algorithm on a classical computer is also not possible due to the exponential growth of the memory required to represent quantum operators and quantum bits in a classical computer. Fig.~\ref{fig:memoryrequirement} illustrates the relationship between the number of qubits and the size of memory used to represent a quantum operator.

\begin{figure}
\centering
\includegraphics[scale=0.5]{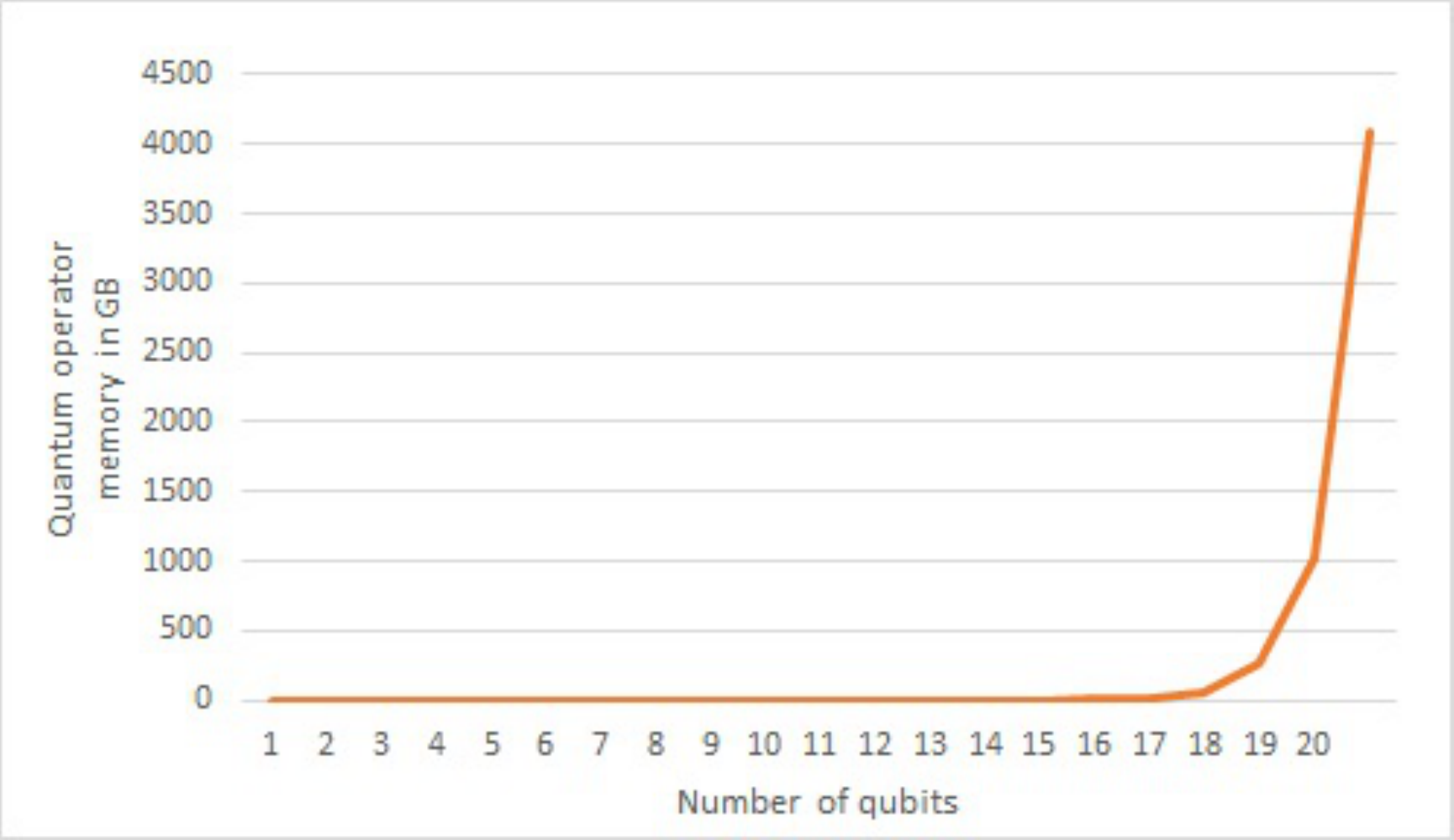}
\caption{Relation between the number of qubits $n$ and size of quantum operators with matricial representation in $\mathbb{R}^{2^n\times 2^n}$ using standard C++ floating point precision}
\label{fig:memoryrequirement}
\end{figure}

To illustrate the impossibility of carrying out empirical analyses of Algorithm \ref{alg:lssa} let us consider the number of qubits necessary to represent a perceptron to learn the Iris dataset\cite{Lichman:2013}. Let $\qop{N}$ be a quantum perceptron over a n-dimensional field $F$, then each attribute of the dataset and each weight of the network will be represented by $n$ quantum bits. Iris database has 4 real entries, then the perceptron will have 4 weights. Excluding auxiliary quantum registers, weights and inputs will be represented by $8n$ quantum bits.
An operator on 8n quantum bits will be represented by a matrix $2^{8n} \times 2^{8n}$. The number of bytes required to represent a  $2^{8n} \times 2^{8n}$ real matrix using the standard C++ floating point data type  is $f(n) = 4\cdot (2^{8n})^2$. 
Note that using only three quantum bits to represent the weights and input data the memory required for simulation is 1024 terabytes. Thus the (quantum or classical) simulation of the learning algorithm in real or synthetic problems is not possible with the current technology.

The multilayer QPF is a generalization of a multilayer classical perceptron and their generalization capabilities are at least equal.  There is an increasing investment in quantum computing by several companies and research institutions to create a general-purpose quantum computer and it is necessary to be prepared to exploit quantum computing power to develop our knowledge in quantum algorithms and models. 

If there are two or more architectures with desired performance in the training set, Algorithm~3 will choose the architecture represented (or addressed) by the string of qubits with more 0’s. This information allows the use of SAL to select the minimal architecture that can learn the training set.

Nonlinear quantum mechanics has been studied since the eighties~\cite{PhysRevLett.62.485} and several neural networks models and learning algorithms used nonlinear quantum computing~\cite{Gupta2001355,panella:09,panella2009neurofuzzy,zhou:12}, but the physical realizability  of nonlinear quantum computing is still controversial~\cite{PhysRevLett.81.3992,raey}. A linear version of SAL needs investigation. The main difficulty is that before step 15 the total probability amplitude of desired  configurations is exponentially smaller than the probability amplitude of undesired configurations. This is an open problem and it may be solved performing some iterations of “classical” learning in the states in the superposition before performing the recovery of the best architecture.

\section{Conclusion}
\label{sec:conclusion}

	We have analysed some models of quantum perceptrons and verified that some of previously defined quantum neural network models in the literature does not respect the principles of quantum computing. Based on this analysis, we presented a new quantum perceptron named quantum perceptron over a field (QPF). The QPF differs from previous models of quantum neural networks since it can be viewed as a direct generalization of the classical perceptron and can be trained by a quantum learning algorithm.

	 We have also defined the architecture of a multilayer QPF and a learning algorithm named Superposition based Architecture Learning algorithm (SAL) that performs a non-linear search in the neural network parameters and the architecture space simultaneously. SAL is based on previous learning algorithms. The main difference of our learning algorithm is the ability to perform a global search in the space of weights and architecture with linear cost  in the number of patterns in the training set and in the number of bits used to represent the neural network. The principle of superposition and a nonlinear quantum operator are used to allow this speedup.

	 The final step of Algorithm~\ref{alg:lssa} is a non-linear search in the architecture and weights space. In this step, free parameters will collapse to a basis state not in superposition. One possible future work is to analyse how one can use the neural network with weights in superposition. In this way, one could take advantage of superposition in a trained neural network.

\section*{Acknowledgments} \noindent  This work is supported by research 
grants from CNPq, CAPES and FACEPE (Brazilian research agencies). We would like to thank the anonymous reviewers for their valuable comments and suggestions to improve the
quality of the paper.

\section*{References}

\bibliography{bibliografia}

\end{document}